\documentclass[12pt,a4paper]{article}
\usepackage[utf8]{inputenc}
\usepackage{amssymb,amsfonts,amsmath,bm}
\usepackage{subfigure}
\usepackage{graphicx}
\usepackage{hyperref}
\usepackage{dsfont}
\usepackage{bbold}
\usepackage{amsthm}
\usepackage{color}
\usepackage{mathtools}
\usepackage{multicol}
\usepackage{tikz}
\usepackage{cite}
\usepackage{float}
\allowdisplaybreaks

\usepackage[
top    = 3.2cm,
bottom = 3.0cm,
left   = 2.8cm,
right  = 2.8cm]{geometry}

\newtheorem{thm}{Theorem}[section]

\newtheorem{remark}[thm]{Remark}
\newtheorem{theorem}[thm]{Theorem}

\DeclareMathOperator{\tr}{tr}

\DeclareMathOperator{\str}{str}
\DeclareMathOperator{\sdet}{sdet}

\begin{document}

\begin{titlepage}

\begin{center}

{\Large \bf Entwining Yang-Baxter maps over Grassmann algebras}

\vskip 1.5cm

{{\bf P. Adamopoulou$^{\star}$ and G. Papamikos$^{\ast}$ }} 

\vskip 0.8cm

{\footnotesize
$^{\star}$ Maxwell Institute for Mathematical Sciences and Department of Mathematics, \\ Heriot-Watt University, Edinburgh, UK}
\\
{\footnotesize
$^{\ast}$ School of Mathematics, Statistics and Actuarial Science, University of Essex, UK}

\vskip 0.5cm

{\footnotesize {\tt E-mail: p.adamopoulou@hw.ac.uk,  g.papamikos@essex.ac.uk }}\\

\end{center}

\vskip 2.0cm

\begin{abstract}
\noindent 
We construct novel solutions to the set-theoretical entwining Yang-Baxter equation. These solutions are birational maps involving  non-commutative dynamical variables which are elements of the Grassmann algebra of order $n$. The maps arise from refactorisation problems of Lax supermatrices associated to a nonlinear Schr\"odinger equation. In this non-commutative setting, we construct a spectral curve associated to each of the obtained maps using the characteristic function of its monodromy supermatrix. We find generating functions of invariants (first integrals) for the entwining Yang-Baxter maps from the moduli of the spectral curves. Moreover, we show that a hierarchy of birational entwining Yang-Baxter maps with commutative variables can be obtained by fixing the order $n$ of the Grassmann algebra. We present the members of the hierarchy in the case $n=1$ (dual numbers) and $n=2$, and discuss their dynamical and integrability properties, such as Lax matrices, invariants, and measure preservation. 
\end{abstract}

\

\hspace{.2cm} \textbf{Mathematics Subject Classification:} 16T25, 15A75, 37J70

\hspace{.2cm} \textbf{Keywords:} Yang--Baxter equations, birational maps, Grassmann algebras, 

\hspace{.2cm} Lax matrices, discrete dynamical systems

\vfill

\end{titlepage}

\section{Introduction}

The first appearances of the  Yang-Baxter (YB) equation can be traced back to the study of quantum many-body systems and exactly solvable models in statistical mechanics \cite{McGuire1964, Yang1967, Baxter}. After that, the YB equation appeared in a broad range of different fields, from quantum field theory and quantum inverse scattering method to gauge theory, and quantum groups. See for example \cite{Jimbo1989, YB_Jimbo}, and references therein, for an introduction and a historical overview of research in relation to the equation at the end of the previous century. Naturally, an intensive focus on finding and classifying solutions to the equation followed \cite{Hietarinta1997, Sklyanin1988}. Originally, the focus was on finding solutions of the YB equation 
\begin{equation}\label{st YB intro}
\mathcal{R}^{12}\mathcal{R}^{13}\mathcal{R}^{23} = \mathcal{R}^{23}\mathcal{R}^{13}\mathcal{R}^{12}
\end{equation}
that are linear maps $\mathcal{R}:V\otimes V\rightarrow V\otimes V$, where $V$ is a $\mathbb{F}$-vector space. Here, $\mathcal{R}^{13}$, for example, denotes the action of $\mathcal{R}$ on the first and third copy of the triple tensor product $V\otimes V\otimes V$. The study of another class of solutions to the YB equation  was proposed by V. Drinfeld in \cite{Drinfeld1992}, where now these solutions are maps $\mathcal{R}:A\times A \rightarrow A\times A$ and $A$ can be any set. Such solutions are often called set-theoretical solutions or Yang-Baxter maps, with the latter term introduced by Veselov in \cite{Veselov2003}. A plethora of works on YB maps has followed since then in relation to  integrable lattice equations \cite{DiscreteBook, PapTongas2007, Nimmo2010, KassotakisNies2012} and the concept of consistency around the cube \cite{Nijhoff2002, BobenkoSuris2002b}, soliton interactions \cite{AblowPrinari2006, Caudrelier2014, dimakis2019, Papamikos2016, CaudrelierGkogkouPrinari2023}, simplex equations \cite{Hietarinta1997, Dimakis2015}, discrete dynamical systems \cite{JoshiKassotakis2019, Veselov2003}, geometric crystals \cite{Etingof2003}, the theory of braces \cite{Rump2007, DoikouSmoktunowicz2023, DoikouRybolowicz2023}, and many other fields. See also \cite{ABS2004},  \cite{PTSV2010} and \cite{Etingof1999} for related classifications. A well-known connection of equation \eqref{st YB intro} is with the braid relation
\begin{equation}\label{Braid}
\mathcal{B}^{12}\mathcal{B}^{13}\mathcal{B}^{12}=\mathcal{B}^{13}\mathcal{B}^{12}\mathcal{B}^{13}\,,
\end{equation}
which is central in knot theory \cite{kasselturaev}. In particular, if $\mathcal{R}$ is a solution to the YB equation \eqref{st YB intro} then $\mathcal{B}=\pi \mathcal{R}$, where $\pi$ is the flip map over $V\otimes V$,  satisfies the braid relation \eqref{Braid}. In recent years, a lot of attention has been focused on the study of solutions to the YB equation over associative algebras, i.e. in not necessarily commutative settings \cite{Doliwa2014, KRK2018, grahovski2016, KRM2016, Kassotakis22, BobenkoSuris2002}. For example, solutions to equation \eqref{st YB intro} over Grassmann algebras were first constructed in \cite{grahovski2016, KRM2016} and later in \cite{KRK2018}. In \cite{AKRP2021}, the authors studied an extension of a well-known YB map, the Adler map, over Grassmann algebras and discussed the Liouville integrability of this map in the case of dual numbers. It is worth mentioning that the study of the integrability of maps over associative but not commutative algebras was proposed as an open problem in \cite{bolsinov2018}.

A generalisation of \eqref{st YB intro} originating in the study of quantum integrable systems, see for example \cite{NijhoffCapelPap1992, Hlavaty1997, AvanDoikouRolletNagy2004}, is given by the following equation
\begin{equation} \label{eYB intro}
\mathcal{S}^{12} \mathcal{R}^{13} \mathcal{T}^{23} = \mathcal{T}^{23 } \mathcal{R}^{13}  \mathcal{S}^{12} \,.
\end{equation}
A triplet of maps $\mathcal{S}$,  $\mathcal{R}$, $\mathcal{T}$ satisfying \eqref{eYB intro} were first derived in \cite{KP2011}, inspired by the work in \cite{Brzezinski2005}, and the name entwining YB equation was introduced for \eqref{eYB intro}. Entwining YB maps have been constructed using certain symmetries of YB maps \cite{kassotakis2019}, and have also appeared in connection to soliton interactions of a matrix KP equation \cite{dimakis2020}, and classical star-triangle relations \cite{kels2023}.  One of the outcomes of the current paper is the construction of a triplet of birational maps satisfying equation \eqref{eYB intro}, which act on noncommutative Grassmann variables. To our knowledge this is the first time that such type of entwining YB maps appear in the literature. 

The structure of the paper is as follows: Section \ref{sec: Prelim} introduces the notation used in the paper, as well as necessary background material. We define the Grassmann algebra $\Gamma(n)$ as well as the appropriate notions of trace, determinant and characteristic rational function for matrices with entries in $\Gamma(n)$. Further, we introduce the parametric entwining YB equation and define the concept of a Lax triple associated to solutions of the equation. In Section \ref{sec: YB maps} we derive novel birational parametric entwining YB maps, which are related to a Grassmann nonlinear Schr\"{o}dinger (NLS) equation. Moreover, we use monodromy supermatrices associated to the Lax triple to construct invariants for each of the maps. In Sections \ref{sec: n1} and \ref{sec: n2} we consider the cases of Grassmann algebras with one and two generators, respectively. In this way we construct the first two members of a hierarchy of entwining YB maps with commutative variables. For each such triplet of maps we derive a Lax representation and several invariants, and we also prove that the maps are measure preserving. Finally, in Section \ref{sec: concl} we offer some concluding remarks in relation to this work, and discuss some directions of future work.

\section{Preliminaries}\label{sec: Prelim}

\subsection{Grassmann algebras} 

We denote by $\Gamma(n)$ the Grassmann algebra of order $n$ over a field $\mathbb{F}$ of characteristic zero (such as $\mathbb{R}$ or $\mathbb{C}$). $\Gamma(n)$  is an associative algebra with unit $1$ and $n$ generators $\theta_i$, $i=1, \ldots, n$, satisfying  
\begin{equation}
    \theta_i \theta_j + \theta_j \theta_i = 0\,.
    \label{eq:defrelat}
\end{equation}
The elements of $\Gamma(n)$ that contain sums of products of only even (resp. odd) number of $\theta_i$'s are called \textit{even} (resp. \textit{odd}) and are denoted by $\Gamma(n)_0$ (resp. $\Gamma(n)_1$).  Even elements commute with all elements of $\Gamma(n)$, while the odd elements anticommute with each other. The Grassmann algebra $\Gamma(n)$, considered as a vector space, can be written as the direct sum of $\Gamma(n)_0$ and $\Gamma(n)_1$, namely $\Gamma(n)=\Gamma(n)_0\oplus\Gamma(n)_1$. Moreover, we have that $\Gamma(n)$ has a natural $\mathbb{Z}_2$-grading, i.e. $\Gamma(n)_i\Gamma(n)_j\subseteq\Gamma(n)_{(i+j)\mod 2}$ and therefore $\Gamma(n)_0$ is a subalgebra of $\Gamma(n)$. In what follows we denote elements of $\Gamma(n)_{0}$  by Latin letters, and elements of $\Gamma(n)_{1}$ by Greek letters, with the exception of $\lambda$ which plays the role of the spectral parameter and takes values in the field $\mathbb{F}$. A Grassmann algebra is an example of a superalgebra, i.e. a super vector space with a $\mathbb{Z}_2$-grading. For more details on superalgebras we direct the reader to \cite{berezin, dictionaryLie, rogers}.

We denote by $\mathbb{F}_n^{k,l}$ the $(k,l)$-dimensional superspace consisting of tuples of $k$ even and $l$ odd variables of a Grassmann algebra of order $n$ over $\mathbb{F}$, namely
\begin{equation}\label{Vn}
    \mathbb{F}_n^{k, l}:= \lbrace (\bm x,\bm \chi) \, | \, \bm x \in \Gamma(n)^{k}_0,~ \bm \chi \in  \Gamma(n)^{l}_1  \rbrace \,.
\end{equation}
In Sections \ref{sec: n1}, \ref{sec: n2} we present examples of maps over Grassmann algebras of order $n=1$ and $n=2$, respectively. Specifically, the $n=1$ case is the algebra of dual numbers (over $\mathbb{F}$), where an element of the algebra is of the form $a+b\theta$ with $\theta^2=0$ and $a,b \in \mathbb{F}$. In the case $n=2$, a generic element of the algebra can be written in the form $a+b\theta_1+c\theta_2+d\theta_1\theta_2$, with $\theta_1, \theta_2$ satisfying \eqref{eq:defrelat}, and with $a+d \theta_1\theta_2 \in \Gamma(2)_0$, $b\theta_1 + c\theta_2 \in \Gamma(2)_1$ and $a,b,c,d\in\mathbb{F}$.

We will be working with square matrices with elements in $\Gamma(n)$ (such matrices are examples of supermatrices), of the block-form 
$$M=\left(
\begin{matrix}
 P & \Pi \\
 \Lambda & L
\end{matrix}\right),$$
where $P$, $L$ are $p \times p$ and $q \times q$ matrices with elements in $\Gamma(n)_0$, while $\Pi$, $\Lambda$ are $p \times q$ and $q \times p$ matrices with elements in $\Gamma(n)_1$. We say that, for example, $P$ is an element of $\mbox{Mat}_p(\Gamma(n)_0)$ and $\Pi$ of $\mbox{Mat}_{p,q}(\Gamma(n)_1)$. We also assume that $\det(L)$ and $\det(P)$ are non-zero. We denote the set of $(p+q) \times (p+q)$ supermatrices, such as $M$, by $\mathrm{M}_{p,q}$. The  supertrace of $M \in \mathrm{M}_{p,q}$ is defined by
\begin{equation} \label{str}
\str(M)=\tr (P)-\tr (L)\,,
\end{equation}
and has the cyclic property $\str(MN)= \str(NM)$. The superdeterminant for $M \in \mathrm{M}_{p,q}$ is defined by:
\begin{equation} \label{sdet}
\sdet(M)=\det(P-\Pi L^{-1}\Lambda)\det(L)^{-1}=\det(P)\det(L-\Lambda P^{-1}\Pi)^{-1},
\end{equation}
and is multiplicative, meaning 
\begin{equation}
\sdet(M_1 M_2) = \sdet(M_1)\sdet(M_2)\,,
\label{eq:multi-det}
\end{equation}
for $M_1, M_2 \in \mathrm{M}_{p,q}$. It follows that the characteristic (rational) function
\begin{equation}
f_{M}(k)=\sdet(M-kI_{p,q})
\label{char fun}
\end{equation}
for a matrix $M \in \mathrm{M}_{p,q}$, with $I_{p,q}$ the unit supermatrix in $\mathrm{M}_{p,q}$, is invariant under similarity transformations $M \to UMU^{-1}$ for $U \in \mathrm{M}_{p,q}$, see \cite{kobayashi1990}. Indeed, 
\begin{equation} \label{inv sim}
\begin{array}{rcl}
f_{UMU^{-1}}(k) & = & \sdet(UMU^{-1}-kI_{p,q})\\
                      & = & \sdet(U(M-kI_{p,q})U^{-1}) = f_M(k)
\end{array}
\end{equation}
where the last equality follows from \eqref{eq:multi-det} and the fact that $\sdet(U^{-1})=\sdet(U)^{-1}$.

\subsection{Parametric entwining Yang-Baxter equation}

In this section we introduce a type of solutions of the entwining YB equation that depend on certain parameters in $\mathbb{F}$. For consistency we also assume that $A$ is $\mathbb{F}^d$ for a positive integer $d$. In many examples in the literature, the field $\mathbb{F}$ is $\mathbb{C}$ and the set $A$ is $\mathbb{C}^d$ or $\mathbb{C}\mathbb{P}^d$ and the resulting (entwining) YB maps are birational isomorphisms of $A\times A$. For more exotic examples see \cite{Doliwa2014, Kassotakis22}.

We consider the maps
$$
\mathcal{R}_{a,b}, \, \mathcal{S}_{a,b}, \, \mathcal{T}_{a,b}: \: A \times A \rightarrow A \times A \,,
$$ 
which depend on parameters $a, b \in \mathbb{F}$. Given for example map $\mathcal{R}_{a,b}$, we denote by $\mathcal{R}_{a,b}^{i,j}$ with $i \neq j \in \lbrace 1,2,3 \rbrace$ the extended map which acts as $\mathcal{R}_{a,b}$ on the $i$ and $j$ copies of the triple Cartesian product of $A$ with itself, and identically on the remaining copy of $A$.  More precisely, we have  
\begin{equation*}
\mathcal{R}_{a,b}^{12}=\mathcal{R}_{a, b}\times id\,, \quad \mathcal{R}_{a,b}^{23}=id\times \mathcal{R}_{a,b} \,,  \quad \mathcal{R}_{a,b}^{13}=\pi^{12} \circ \mathcal{R}_{a,b}^{23} \circ \pi^{12}\,,
\end{equation*}
with $\pi^{12}$ the extension of the permutation (flip) map $\pi:(x,y)\to (y,x)$ on $A \times A$, and $id$ the identity map on $A$. In this paper we will be concerned with parametric ordered triplets of maps $(\mathcal{S}_{a,b}, \mathcal{R}_{a,b}, \mathcal{T}_{a,b})$ which satisfy the parametric entwining YB equation
\begin{equation}\label{eYB_eq1}
\mathcal{S}^{12}_{a,b}\circ \mathcal{R}^{13}_{a,c} \circ \mathcal{T}^{23}_{b,c} = \mathcal{T}^{23}_{b,c}\circ \mathcal{R}^{13}_{a,c} \circ \mathcal{S}^{12}_{a,b}\, .
\end{equation}
We call such maps entwining YB maps. Equation \eqref{eYB_eq1} is to be understood as equality of compositions of maps over the triple product $A\times A\times A$.

The entwining YB equation \eqref{eYB_eq1} can be represented by the following diagram, where the lines are coloured to indicate that each crossing corresponds to a different map, e.g. the red-blue crossing corresponds to $\mathcal{S}^{12}_{a,b}$, etc. 
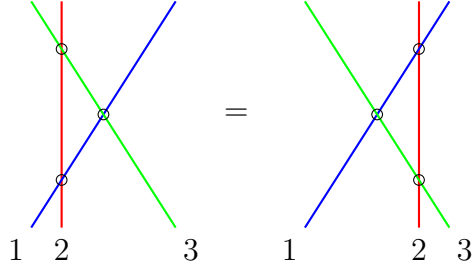
\begin{figure}[H]
\centering
\begin{tikzpicture}
\draw[red, thick] (0,0) -- (0,3);
\draw[blue, thick] (-0.4,0) -- (1.5,3);
\draw[green, thick] (-0.4,3)--(1.5,0);
\draw (0,12/19) circle (2pt);
\draw (0,45/19) circle (2pt);
\draw (11/20,3/2) circle (2pt);
\draw (3/2+0.8,3/2) node {=};
\draw[red, thick] (1.5+3.2,0) -- (3.2+1.5,3);
\draw[blue, thick] (3.2,0) -- (3.2+1.9,3);
\draw[green, thick](3.2,3)--(3.2+1.9,0);
\draw (4.15,1.5) circle (2pt);
\draw (4.15+11/20,12/19) circle (2pt);
\draw (4.15+11/20,45/19) circle (2pt);
\draw (-0.6,-0.3) node {1};
\draw (0,-0.3) node {2};
\draw (1.7,-0.3) node {3};
\draw (3,-0.3) node {1};
\draw (4.7,-0.3) node {2};
\draw (1.9+3.4,-0.3) node {3};
\end{tikzpicture}
\caption{Diagrammatic representation of entwining YB equation \eqref{eYB_eq1}}
\end{figure}
\noindent  When all lines have the same colour, i.e. when $\mathcal{S}_{a,b}=\mathcal{R}_{a,b}=\mathcal{T}_{a,b}$,  then equation \eqref{eYB_eq1} reduces to the parametric YB equation
\begin{equation} \label{YBE}
\mathcal{R}^{12}_{a,b}\circ \mathcal{R}^{13}_{a,c} \circ \mathcal{R}^{23}_{b,c} = \mathcal{R}^{23}_{b,c}\circ \mathcal{R}^{13}_{a,c} \circ \mathcal{R}^{12}_{a,b}\,.
\end{equation}
In general, a parametric YB map $R_{a,b}(x,y)=(u_{a,b}(x,y),v_{a,b}(x,y))$ is called non-degenerate if the maps $u_{a,b}(\cdot, y):A\to A$ and $v_{a,b}(x,\cdot):A\to A$ are bijective \cite{Etingof1999, Etingof2003}. More recently, non-degenerate YB maps which are also birational have been referred to as quadrirational YB maps \cite{ABS2004, PTSV2010}. We use the same terminology for maps that satisfy the entwining YB equation \eqref{eYB_eq1}.

Following \cite{KP2011}, we define a strong Lax triple for maps  $\mathcal{S}_{a,b}, \mathcal{R}_{a,b}, \mathcal{T}_{a,b}$ to be a set of three matrices $(\mathcal{L}_a, \mathcal{M}_a, \mathcal{N}_a)$, each depending on a point $x\in A$, a parameter $a\in\mathbb{F}$ and a spectral parameter $\lambda \in \mathbb{F}$,  such that the matrix refactorisation problems 
\begin{subequations}\label{Lax 3ple}
\begin{align}
\mathcal{L}_a(u) \mathcal{M}_b(v) &=\mathcal{M}_b(y) \mathcal{L}_a(x)\,, \label{S refact}\\
\mathcal{L}_a(u) \mathcal{N}_b(v) &=\mathcal{N}_b(y) \mathcal{L}_a(x)\,, \label{R refact}\\
\mathcal{M}_a(u) \mathcal{N}_b(v) &=\mathcal{N}_b(y) \mathcal{M}_a(x)\,, \label{T refact}
\end{align}
\end{subequations}
imply uniquely the maps $\mathcal{S}_{a,b}, \mathcal{R}_{a,b}, \mathcal{T}_{a,b}: (x,y) \to (u,v) $, respectively.  If equations \eqref{S refact}-\eqref{T refact} are satisfied for given $\mathcal{S}_{a,b}, \mathcal{R}_{a,b}, \mathcal{T}_{a,b}$ maps, then the triple of matrices is called simply a Lax triple. In general we omit the dependence of the Lax matrices on the spectral parameter $\lambda$ for convenience. It was proved in \cite{KP2011} that if $(\mathcal{L}_a,\mathcal{M}_a,\mathcal{N}_a)$ is a strong Lax triple for maps $\mathcal{S}_{a,b}, \mathcal{R}_{a,b}, \mathcal{T}_{a,b}$ and the following equality
\begin{equation}\label{trifact}
\mathcal{L}_a(x) \mathcal{M}_b  (y) \mathcal{N}_c (z) = \mathcal{L}_a(x') \mathcal{M}_b  (y') \mathcal{N}_c (z') 
\end{equation}
implies that $x = x', \, y = y' $ and $z = z' $ then the maps are entwining YB maps.  If $\mathcal{L}_{a}=\mathcal{M}_{a}=\mathcal{N}_{a}$, then the refactorisation problems \eqref{S refact}-\eqref{T refact} coincide and $\mathcal{L}_{a}$ is a strong Lax matrix for the parametric YB map $\mathcal{R}_{a,b}$.

The trace of products of the Lax matrices in \eqref{Lax 3ple} gives invariants of the maps $\mathcal{S}_{a,b}, \mathcal{R}_{a,b}, \mathcal{T}_{a,b}$. By invariant of a map, say $\mathcal{R}_{a,b}$, we mean a function $I$ such that $I \circ \mathcal{R}_{a,b} = I$. Moreover, an anti-invariant $I$ is a function such that $I \circ \mathcal{R}_{a,b} = - I$. It follows that the product of two different anti-invariants or the square of anti-invariants are all invariants of the given map.

In Section \ref{sec: YB maps} we  derive maps  $\mathcal{S}_{a,b}, \mathcal{R}_{a,b}, \mathcal{T}_{a,b}$ satisfying the parametric entwining YB equation \eqref{eYB_eq1} over the Grassmann algebra $\Gamma(n)$.   In this case, the set $A$ is the $(k,l)$-dimensional superspace $\mathbb{F}^{k,l}_n$ for given positive integers $k,l$, and the obtained entwining YB maps are birational automorphisms of $\mathbb{F}^{k,l}_n\times \mathbb{F}^{k,l}_n$. Similarly, the Lax triple $(\mathcal{L}_a, \mathcal{M}_a, \mathcal{N}_a)$ will depend on a point $(\bm{x},\bm{\chi})\in \mathbb{F}^{k,l}_n$. The maps we construct are Grassmann extensions of the following entwining YB maps
\begin{equation}\label{YB KRP}
\begin{aligned}
{\displaystyle \mathtt{S}_{a,b}(x_1,x_2,y_1,y_2)} &= {\displaystyle \left( \frac{y_1^2}{bx_1} + \frac{y_1}{b}(y_2-a), \frac{b}{y_1}, x_1, a+x_1x_2 -\frac{y_1}{x_1}    \right)}\,, \\
{\displaystyle \mathtt{R}_{a,b}(x_1,x_2,y_1,y_2)} &= {\displaystyle \left( y_1 - \frac{a-b}{1+x_1y_2}x_1, y_2,x_1,x_2 + \frac{a-b}{1+x_1y_2}y_2  \right)} \,, \\
{\displaystyle \mathtt{T}_{a,b}(x_1,x_2,y_1,y_2)} &= {\displaystyle \left( \frac{a}{y_2}, b+ y_1y_2 - \frac{a}{x_1y_2},x_1, \frac{a}{x_1^2y_2} +\frac{x_2-b}{x_1}  \right)} \,.
\end{aligned}
\end{equation}
Maps \eqref{YB KRP} were derived in \cite{KRP2019}, and were shown to admit a strong Lax triple and to be Liouville integrable having polynomial and rational invariants which Poisson-commute.

\section{Grassmann entwining YB maps} \label{sec: YB maps}

In this section we derive birational, parametric, entwining YB maps over the Grassmann algebra $\Gamma(n)$, starting from the refactorisation problems of certain Lax supermatrices. These Lax matrices are Darboux matrices associated to an NLS equation. Darboux matrices for NLS type equations were derived in \cite{KRMX2015}, while in \cite{grahovski2013} a Grassmann generalisation of the NLS and associated Darboux transformations were presented.  Refactorisation problems of certain Darboux matrices over Grassmann algebras were considered in \cite{grahovski2016}, resulting to YB maps with Grassmann variables. Moreover, in \cite{KRP2019} the parametric entwining YB maps given in \eqref{YB KRP} were derived from the refactorisation problems of the Darboux matrices which were presented in \cite{KRMX2015}. The resulting entwining YB maps of this section are generalisations of \eqref{YB KRP} involving non-commutative (Grassmann) variables. We note here that while the maps that we obtain in this paper are birational, they are degenerate or non-quadrirational.

We consider the following supermatices in $\mathrm{M}_{2,1}$
\begin{equation} \label{Lax}
\mathcal{L}_a({\bm x}, {\bm \chi}) = \begin{pmatrix}
x_1x_2+ \chi_1 \chi_2 + a + \lambda & x_1 & \chi_1\\
x_2 & 1& 0\\
\chi_2 & 0 & 1
\end{pmatrix} \!, \;
\mathcal{M}_a({\bm x}, {\bm \chi}) = \begin{pmatrix}
x_2 + \lambda & x_1 & \chi_1\\
\frac{a}{x_1} & 0& 0\\
\chi_2 & 0 & 1
\end{pmatrix} \!,
\end{equation}
with $({\bm x}, {\bm \chi}) =(x_1, x_2, \chi_1, \chi_2)  \in \mathbb{F}_n^{2,2}$, $a \in \mathbb{F}$ a parameter, and $\lambda \in \mathbb{F}$ a spectral parameter. The $2 \times 2$ blocks of $\mathcal{L}_a$ and $\mathcal{M}_a$ with entries in $\Gamma(n)_0$  constitute the Darboux matrices for NLS derived in \cite{KRMX2015}. The refactorisation problems \eqref{S refact}-\eqref{T refact} for matrices \eqref{Lax}, with $\mathcal{N}_a \equiv \mathcal{L}_a$, have unique solutions for $\bigl( ({\bm u}, {\bm \xi}), ({\bm v}, {\bm \eta}) \bigr)$ in terms of $\bigl( ({\bm x}, {\bm \chi}), ({\bm y}, {\bm \psi}) \bigr)$. These give rise to eight-dimensional birational maps $\mathcal{R}_{a,b}, \, \mathcal{S}_{a,b}, \,\mathcal{T}_{a,b}$ with even-odd Grassmann variables which act as 
$$
 \left( (x_1, x_2, \chi_1, \chi_2), (y_1, y_2, \psi_1, \psi_2)  \right) \mapsto \left( (u_{1}, u_{2}, \xi_{1}, \xi_{2}), (v_{1}, v_{2}, \eta_{1}, \eta_{2})  \right)\,.
$$
In particular, map $\mathcal{R}_{a,b}$ is defined by the following expressions 
\begin{equation} \label{R comps}
\mathcal{R}_{a,b}:
\begin{cases}
 u_1 = y_1 - \frac{(a-b)(1+x_1y_2-\chi_1 \psi_2)}{(1+x_1y_2)^2} x_1 \,,     &    v_1 = x_1\,,  \\
  u_2 =  y_2 \,,      &    v_2  =  x_2 +  \frac{(a-b)(1+x_1y_2- \chi_1 \psi_2)}{(1+x_1y_2)^2} y_2 \,,  \\
 \xi_1 = \psi_1 -  \frac{a-b }{1+ x_1y_2} \chi_1 \,,  &    \eta_1 = \chi_1 \,, \\
 \xi_2 = \psi_2 \,,  &    \eta_2 =  \chi_2 +  \frac{a-b}{1+x_1y_2} \psi_2 \,, 
\end{cases}
\end{equation}
map $\mathcal{S}_{a,b}$ is given by 
\begin{equation} \label{S comps}
\! \mathcal{S}_{a,b}:
\begin{cases}
 u_1 = \frac{y_1^2 (1+ \chi_1 \psi_2)}{b x_1}  +  \frac{y_1 (y_2-a-\psi_1 \psi_2)}{b}\,,     &   v_1 = x_1\,, \vspace{0.2cm} \\
 u_2 =  \frac{b}{y_1} \,,      &   v_2 =  a+x_1 x_2 - \frac{y_1}{x_1} + \chi_1 \chi_2 \,, \\
 \xi_1 = \psi_1 - \frac{y_1}{x_1} \chi_1\,,  &   \eta_1 = \chi_1 \,, \\
 \xi_2 = \psi_2 \,,  &   \eta_2 = \chi_2 +\frac{y_1}{x_1} \psi_2 \,, 
\end{cases}
\end{equation}
and finally map $\mathcal{T}_{a,b}$ is
\begin{equation} \label{T comps}
\mathcal{T}_{a,b}:
\begin{cases}
u_1 = \frac{a}{y_2} \,,     &  v_1 = x_1\,,  \\
 u_2 = b - \frac{a}{x_1 y_2} + y_1 y_2 + \psi_1 \psi_2 \,,   &  v_2 = \frac{a (1+\chi_1 \psi_2)}{x_1^2 y_2}  + \frac{(x_2-b-\chi_1\chi_2)}{x_1} \,,  \vspace{0.15cm} \\
 \xi_1 = \psi_1 + \frac{a \chi_1}{x_1y_2}\,,    &  \eta_1 = \chi_1 \,,  \\
\xi_2 =  \psi_2 \,,     &   \eta_2 = \chi_2 - \frac{a \psi_2}{x_1 y_2}\,.
\end{cases}
\end{equation}

\begin{theorem} \label{RST maps}
The maps $\mathcal{S}_{a,b}, \, \mathcal{R}_{a,b}, \, \mathcal{T}_{a,b}: \mathbb{F}_n^{2, 2} \times \mathbb{F}_n^{2, 2}  \rightarrow \mathbb{F}_n^{2, 2} \times \mathbb{F}_n^{2, 2}$ defined in \eqref{R comps}-\eqref{T comps} admit a strong Lax triple  $(\mathcal{L}_a, \mathcal{M}_a, \mathcal{N}_a)$, with $\mathcal{L}_a, \mathcal{M}_a$ given in \eqref{Lax} and $\mathcal{N}_a \equiv \mathcal{L}_a$, and they satisfy the parametric entwining Yang--Baxter equation \eqref{eYB_eq1}.
\end{theorem}

\begin{proof}
The first part of the proof can be shown directly by solving the following refactorisation problems 
\begin{subequations}\label{Lax 3ple - Gr}
\begin{align}
\mathcal{L}_a({\bm u},{\bm \xi}) \mathcal{M}_b({\bm v},{\bm \eta}) &=\mathcal{M}_b({\bm y},\bm{\psi}) \mathcal{L}_a({\bm x},{\bm \chi})\,, \label{S refact - Gr}\\
\mathcal{L}_a({\bm u},{\bm \xi}) \mathcal{L}_b({\bm v},{\bm \eta}) &=\mathcal{L}_b({\bm y},\bm{\psi}) \mathcal{L}_a({\bm x},{\bm \chi})\,, \label{R refact - Gr}\\
\mathcal{M}_a({\bm u},{\bm \xi}) \mathcal{L}_b({\bm v},{\bm \eta}) &=\mathcal{L}_b({\bm y},\bm{\psi}) \mathcal{M}_a({\bm x},{\bm \chi})\,, \label{T refact - Gr}
\end{align}
\end{subequations}
and showing that they admit a unique solution for $( {\bm u}, {\bm \xi}, {\bm v}, {\bm \eta} )$  in terms of  $( {\bm x}, {\bm \chi}, {\bm y}, {\bm \psi} )$. The proof that the obtained maps satisfy the entwining YB equation is given in Appendix \ref{proof EYB}.
\end{proof}

\begin{remark} \label{ST equiv}
Maps $\mathcal{S}_{a,b}$ and $\mathcal{T}_{a,b}$ are related to each other by 
$$
\mathcal{T}_{a,b}=\tilde{\mathcal{S}}_{a,b}:=\pi \circ \mathcal{S}_{b,a}^{-1} \circ \pi,
$$
where now $\pi$ is the flip map in $\mathbb{F}^{2,2}_n\times \mathbb{F}^{2,2}_n $, acting as $\pi((\bm x, \bm \chi), (\bm y, \bm \psi)) = ((\bm y, \bm \psi), (\bm x, \bm \chi))$.
This can be readily deduced by observing that the refactorisation problems \eqref{S refact - Gr} and \eqref{T refact - Gr} are related by the transformation $a \leftrightarrow b$, $(\bm u, \bm \xi) \leftrightarrow (\bm y, \bm \psi)$, $(\bm x, \bm \chi) \leftrightarrow (\bm v, \bm \eta)$. This shows the birationality of maps $\mathcal{S}_{a,b}$ and $\mathcal{T}_{a,b}$. Moreover, the invariance of \eqref{R refact - Gr} under the above transformation shows that $\mathcal{R}_{a,b}$ is also a birational map. In what follows we will only consider the maps $\mathcal{R}_{a,b}$ and $\mathcal{S}_{a,b}$.
\end{remark}

Map $\mathcal{R}_{a,b}$ is an extension over Grassmann algebras of the Adler-Yamilov map \cite{adleryamilov1994}. This map and its corresponding matrix refactorisation problem \eqref{R refact - Gr}  were studied in \cite{grahovski2016}, where it was shown  that $\mathcal{R}_{a,b}$ is a birational Yang-Baxter map, and also reversible i.e. it satisfies the relation $\tilde{\mathcal{R}}_{a,b} \circ \mathcal{R}_{a,b} = id$. Taking the commutative limit in maps $\mathcal{S}_{a,b}$ and $\mathcal{T}_{a,b}$, i.e. sending all the odd variables to zero, we obtain the birational maps which were derived in \cite{KRP2019}. Therefore, the maps $\mathcal{S}_{a,b}, \mathcal{R}_{a,b}, \mathcal{T}_{a,b}$ in Theorem \ref{RST maps} form a generalisation over Grassmann algebras of the entwining Yang-Baxter maps given in \eqref{YB KRP}. We show that invariants of these maps can be obtained using the invariance of the superdeterminant under similarity transformations, see \eqref{inv sim} in Section \ref{sec: Prelim}.

\begin{theorem}
The maps $\mathcal{R}_{a,b}, \mathcal{S}_{a,b}:\mathbb{F}^{2,2}_n\times \mathbb{F}^{2,2}_n \to \mathbb{F}^{2,2}_n\times \mathbb{F}^{2,2}_n $, given in \eqref{R comps}-\eqref{S comps}, admit the following $\mathcal{I,J}$-sets of invariants, respectively:
\begin{equation} \label{R inv}
\begin{aligned}
\mathcal{I}_1 &= x_1x_2 + y_1 y_2 \,, \quad \mathcal{I}_2 = \chi_1\chi_2 + \psi_1\psi_2 \,, \quad \mathcal{I}_3 = (x_1 \psi_1 - y_1\chi_1)(x_2\psi_2 - y_2 \chi_2) \,, \\
\mathcal{I}_4 &= b(x_1x_2+\chi_1\chi_2) + a(y_1y_2+\psi_1\psi_2) +  y_1y_2 (x_1x_2+\chi_1\chi_2)  + x_1x_2\psi_1\psi_2  \\ & + x_1y_2 + x_2y_1 + \chi_1 \psi_2 + \psi_1 \chi_2\,, \quad \mathcal{I}_5 =   \chi_1\chi_2\psi_1\psi_2\,,
  \end{aligned}
\end{equation}
\begin{equation} \label{S inv}
\begin{aligned}
\mathcal{J}_1 &= y_2 + x_1x_2 + \chi_1\chi_2 \,, \quad \mathcal{J}_2 = \chi_1\chi_2 + \psi_1\psi_2 \,, \\
\mathcal{J}_3 &= (a+x_1x_2)\psi_1 \psi_2 + (b+y_2) \chi_1 \chi_2 + (1-x_2y_1)\chi_1 \psi_2 + \left( 1-\frac{b x_1}{y_1} \right)\psi_1  \chi_2 \,, \\
\mathcal{J}_4 &= b \frac{x_1}{y_1}  + y_2(a+ x_1x_2 + \chi_1\chi_2) +  x_2y_1 + \chi_1\psi_2 +\psi_1 \chi_2\,, \quad \mathcal{J}_5 = \chi_1\chi_2\psi_1\psi_2\,.
\end{aligned}
\end{equation} 
\end{theorem}

\begin{proof}
The invariants of map $\mathcal{S}_{a,b}$ are obtained using the monodromy supermatrix  $P_{\mathcal{S}}(\bm{x}, \bm{\chi}, \bm{y}, \bm{\psi}) = \mathcal{M}_b(\bm{y}, \bm{\psi}) \mathcal{L}_a(\bm{x}, \bm{\chi})$, with $\mathcal{L}_a, \mathcal{M}_b$ given in \eqref{Lax}. From the refactorisation property \eqref{S refact - Gr} we obtain the isospectrality property of the monodromy under the action of the map
\begin{equation}
P_{\mathcal{S}}({\bm u}, {\bm \xi}, {\bm v}, {\bm \eta}) = \mathcal{M}_b({\bm v}, {\bm \eta}) P_{\mathcal{S}}({\bm x}, {\bm \chi}, {\bm y}, {\bm \psi}) \mathcal{M}_b^{-1}({\bm v}, {\bm \eta}) \,,
\end{equation}
similarly to the commutative setting. 
It follows that the characteristic (rational) function of the monodromy supermatrix $f_{P_{\mathcal{S}}}(k)=\sdet (P_{\mathcal{S}}-k I_{2,1})$ generates invariants of the map $\mathcal{S}_{a,b}$.

The supermatrix $P_{\mathcal{S}} - k I_{2,1}$ can be written in the form 
\begin{equation}
P_{\mathcal{S}}-k I_{2,1} = \begin{pmatrix}
{\rm A}-kI & {\rm B} \\
{\rm C} & {\rm D}(k)
\end{pmatrix},
\end{equation}
with ${\rm D}(k) = \psi_2\chi_1 + 1-k$ and ${\rm A}, {\rm B}, {\rm C}$ functions of ${\bm x}, {\bm \chi}, {\bm y}, {\bm \psi}$ and $\lambda$. 
Hence, from \eqref{sdet} we have that
$$
f_{P_{\mathcal{S}}}(k,\lambda) =  \frac{\det({\rm A}-kI-{\rm B D}(k)^{-1}{\rm C})}{\det {\rm D}(k)} = \frac{1-k-\psi_2\chi_1}{(1-k)^2} \det({\rm A}- kI - {\rm B D}(k)^{-1}{\rm C})\,.
$$
Setting $\mu=1-k$, and factoring a $\mu^{-2}$ factor outside the determinant, the characteristic function takes the form
\begin{equation}\label{p(s)}
f_{P_{\mathcal{S}}}(\mu, \lambda) = \frac{\mu+\chi_1\psi_2}{\mu^6} \det(\mu^3 I + \mu^2({\rm A}-I) -\mu {\rm B C} - \chi_1\psi_2 {\rm B C}) \,.
\end{equation}
The equation $\mu^3 f_{P_{\mathcal{S}}}(\mu, \lambda)=0$  defines the spectral curve associated to map $\mathcal{S}_{a,b}$ and its moduli provides the $\mathcal{J}$-set of invariants of the map. Indeed, expanding the determinant and using the explicit forms of $A, {\rm B, C}$, we obtain
\begin{equation} \label{spectral poly}
\mu^3f_{P_S}(\mu, \lambda) = \mu^4 + \mu^3f_3(\lambda) + \mu^2f_2(\lambda) +\mu f_1(\lambda) +f_0(\lambda) \,,
\end{equation}
where the $f_i(\lambda)$, with $i=0,\ldots, 3$, are generating functions of the invariants of  map $\mathcal{S}_{a,b}$ and have the following form 
\begin{equation}
\begin{aligned}
f_3(\lambda) &=\lambda^2 + \lambda(\mathcal{J}_1+a) + \mathcal{J}_4 \,,\\
f_2(\lambda) &= -\lambda (\mathcal{J}_1 + \mathcal{J}_2 +a+b) - (2\mathcal{J}_5 + \mathcal{J}_4 + \mathcal{J}_3 +ab) \,,\\
f_1(\lambda) &= \mathcal{J}_3 \,,\\
f_0(\lambda) &= -2 \mathcal{J}_5 \,.
\end{aligned}
\end{equation}

Similarly, we define the monodromy matrix of $\mathcal{R}_{a,b}$ to be $P_{\mathcal{R}}(\bm{x}, \bm{\chi}, \bm{y}, \bm{\psi}) = \mathcal{L}_b(\bm{y}, \bm{\psi}) \mathcal{L}_a(\bm{x}, \bm{\chi})$ and then it follows that the characteristic function $g_{P_{\mathcal{R}}}(\mu, \lambda)$ associated with the map $\mathcal{R}_{a,b}$ can be written in the form
$$
\mu^3 g_{P_{\mathcal{R}}}(\mu, \lambda) = \mu^4 + \mu^3 g_3(\lambda) + \mu^2 g_2(\lambda) + \mu g_1(\lambda) \,,
$$
with 
\begin{equation}
\begin{aligned}
g_3(\lambda) &=\lambda^2 + \lambda( \mathcal{I}_1 + \mathcal{I}_2 +a+b) + \mathcal{I}_5 + \mathcal{I}_4 + ab-1 \,,\\
g_2(\lambda) &= - \lambda (\mathcal{I}_1 + \mathcal{I}_2) - (3\mathcal{I}_5 + \mathcal{I}_4 + \mathcal{I}_3) \,,\\
g_1(\lambda) &= 2\mathcal{I}_5 + \mathcal{I}_3 \, ,
\end{aligned}
\end{equation}
thus obtaining the $\mathcal{I}-$set of invariants of $\mathcal{R}_{a,b}$.
\end{proof}

\begin{remark}
The invariants $\mathcal{I}_1, \mathcal{I}_2, \mathcal{I}_4, \mathcal{I}_5$ of map $\mathcal{R}_{a,b}$ were derived in \cite{grahovski2016} using the supertrace of the monodromy. We notice that $\mathcal{I}_2$ and $\mathcal{I}_5$ are related by $2\mathcal{I}_5 = \mathcal{I}_2^2$. Here we obtain the new invariant $\mathcal{I}_3$ of the map using the characteristic function. 
\end{remark}

\begin{remark} \label{rem: anti-inv}
One can verify that the quantities $\chi_1 \psi_1$ and $\chi_2 \psi_2$ are anti-invariants of all maps $\mathcal{R}_{a,b}, \mathcal{S}_{a,b}, \mathcal{T}_{a,b}$, and that the invariant $\mathcal{I}_5$ (and $\mathcal{J}_5$) can be obtained from the product of those anti-invariants. Moreover, for map $\mathcal{R}_{a,b}$ the quantities $x_i \psi_i - y_i \chi_i$ for $i=1, 2$ are anti-invariants, and $\mathcal{I}_3$ is the product of these two anti-invariants.
\end{remark}

\begin{remark}
Using Remark \ref{ST equiv} we deduce that the invariants of map $\mathcal{T}_{a,b}$ can be obtained from the $\mathcal{J}$-set using the reflection  $a \leftrightarrow b$, ${\bm x} \leftrightarrow {\bm y}$, ${\bm \chi} \leftrightarrow {\bm \psi}$.
\end{remark}

In the following two sections we derive entwining YB maps in dimensions $8$ and $16$ with commutative variables. To achieve this, we start from the derived maps $\mathcal{R}_{a,b}, \mathcal{S}_{a,b}, \mathcal{T}_{a,b}$ and consider the Grassmann algebra $\Gamma(n)$ when $n=1$ and $n=2$. This way, we demonstrate how the first two members of a hierarchy of $2^{n+2}-$dimensional birational entwining YB maps can be obtained for any $n=1, 2, \ldots$.

\section{The Grassmann algebra $\Gamma(1)$} \label{sec: n1}

The algebra $\Gamma(1)$ has unit $1$ and one generator $\theta$ with $\theta ^2 = 0$. This is the case of dual numbers (over $\mathbb{F}$). We consider the maps $\mathcal{R}_{a,b}, \mathcal{S}_{a,b}, \mathcal{T}_{a,b}$ in \eqref{R comps}-\eqref{T comps}  over the Grassmann algebra $\Gamma(1)$, and in this way we derive birational maps $S_{a,b}, \, R_{a,b}, \,T_{a,b}$ with commutative variables which satisfy the entwining YB equation \eqref{eYB_eq1}. Moreover, we study  integrability properties of the obtained maps, such as Lax representation, invariants, and measure preservation.

We expand all variables of the maps $\mathcal{R}_{a,b}, \mathcal{S}_{a,b}, \mathcal{T}_{a,b}$  in \eqref{R comps}-\eqref{T comps} and their images in terms of  $1$ and $\theta$. The Latin variables are in $\Gamma(1)_0$ and are therefore proportional to $1$, while the Greek variables are in $\Gamma(1)_1$ and thus proportional to $\theta$. Comparing coefficients of $1$ and $\theta$ on both sides of equations \eqref{R comps}-\eqref{T comps} we obtain  maps $R_{a,b}, \, S_{a,b}, \,T_{a,b}: \mathbb{F}^8 \rightarrow \mathbb{F}^8$ with 
$$
(x_1,x_2,\chi_1,\chi_2,y_1,y_2,\psi_1,\psi_2)\mapsto (u_1,u_2,\xi_1,\xi_2,v_1,v_2,\eta_1,\eta_2)\,.
$$ 
These maps are given by the following expressions
\begin{equation} \label{R G1 comps}
R_{a,b}:  
\left\lbrace
\begin{aligned}
 u_{1} &= y_1- \frac{a-b}{1+x_1y_2}x_1 \,,  &  v_{1} &= x_{1}\,,  \\
 u_{2} &= y_2 \,,  &  v_{2} &=  x_2 + \frac{a-b}{1+x_1y_2}y_2  \,,   \\
 \xi_{1} &= \psi_{1} - \frac{a-b}{1+x_{1}y_{2}} \chi_1\,,   &  \eta_{1} &= \chi_{1} \,,  \\
 \xi_{2} &=  \psi_{2} \,,  &  \eta_{2} &= \chi_{2} + \frac{a-b}{1+x_1 y_2}\psi_2  \,,
\end{aligned}
\right.
\end{equation}

\begin{equation} \label{S G1 comps}
\hspace{-0.3cm} S_{a,b}:  
\left\lbrace
\begin{aligned}
 u_1 &= \frac{y_1^2}{b x_1} +\frac{ y_1(y_2-a)}{b} \,,  &   v_1 &= x_1\,,  \\
 u_2 &=  \frac{b}{y_1} \,, &  v_2 &=  a+x_1 x_2 - \frac{y_1}{x_1}  \,,\\
 \xi_1 &= \psi_1 - \frac{y_1}{x_1} \chi_1\,,   &   \eta_1 &= \chi_1 \,,\\
 \xi_2 &=  \psi_2  \,, & \eta_2 &= \chi_2 +\frac{y_1}{x_1} \psi_2  \,, 
\end{aligned}
\right.
\end{equation}

\begin{equation} \label{T G1 comps}
T_{a,b}:  
\left\lbrace
\begin{aligned}
u_{1} &= \frac{a}{y_{2}} \,,   &   v_{1} &= x_{1}\,,  \\
u_{2} &= b - \frac{a}{x_{1} y_{2}} + y_{1} y_{2}  \,,  &  v_{2} &=  \frac{a+ x_{1}y_{2}(x_{2}-b)}{x_{1}^2 y_{2}}  \,,   \\
 \xi_{1} &= \psi_{1} + \frac{a \chi_{1}}{x_{1}y_{2}}\,,    &   \eta_{1} &= \chi_{1} \,,  \\
 \xi_{2} &=  \psi_{2} \,,   &  \eta_{2} &= \chi_{2} - \frac{a \psi_{2}}{x_{1} y_{2}}  \,. 
\end{aligned}
\right.
\end{equation}
For simplicity, we have used the same letters for the variables in \eqref{R G1 comps}-\eqref{T G1 comps} as in the case of $\Gamma(n)$. From Theorem \ref{RST maps} it follows that the above maps satisfy the entwining YB equation \eqref{eYB_eq1}.

The eight-dimensional birational maps $R_{a,b}, S_{a,b}, T_{a,b}$ defined by \eqref{R G1 comps}-\eqref{T G1 comps} admit a strong Lax triple $(L_a, M_a, N_a)$,  with $N_a \equiv L_a$ and $L_a, M_a$ given by
\begin{equation}\label{L G1 tensor}
L_a(x_1,x_2,\chi_1,\chi_2)=
\begin{pmatrix}
x_1x_2+a+\lambda & 0&x_1&0&0&\chi_1\\
0&x_1x_2+a+\lambda&0&x_1&0&0\\
x_2&0&1&0&0&0\\
0&x_2&0&1&0&0\\
0&\chi_2&0&0&1&0\\
0&0&0&0&0&1
\end{pmatrix}
\end{equation}
and
\begin{equation} \label{M G1 tensor}
M_a(x_1,x_2,\chi_1,\chi_2)=
\begin{pmatrix}
x_2+\lambda&0&x_1&0&0&\chi_1\\
0&x_2+\lambda&0&x_1&0&0\\
\frac{a}{x_1}&0&0&0&0&0\\
0&\frac{a}{x_1}&0&0&0&0\\
0&\chi_2&0&0&1&0\\
0&0&0&0&0&1
\end{pmatrix},
\end{equation}
where $x_i, \chi_i \in \mathbb{F}$ for $i=1,2$. Namely, each of the matrix refactorisations 
\begin{subequations}\label{G1 refact}
\begin{align}
L_a(u_1,u_2,\xi_1,\xi_2) L_b(v_1,v_2,\eta_1,\eta_2) &=L_b(y_1,y_2,\psi_1,\psi_2)L_a(x_1,x_2,\chi_1,\chi_2)\,, \label{R G1 refact}\\
L_a(u_1,u_2,\xi_1,\xi_2) M_b(v_1,v_2,\eta_1,\eta_2) &=M_b(y_1,y_2,\psi_1,\psi_2)L_a(x_1,x_2,\chi_1,\chi_2)\,, \label{S G1 refact}\\
M_a(u_1,u_2,\xi_1,\xi_2) L_b(v_1,v_2,\eta_1,\eta_2) &=L_b(y_1,y_2,\psi_1,\psi_2)M_a(x_1,x_2,\chi_1,\chi_2)\,,\label{T G1 refact}
\end{align}
\end{subequations}
leads to a system of polynomial equations which can be solved uniquely for $(u_1,u_2,\xi_1,\xi_2,v_1,v_2,\eta_1,\eta_2)$  leading to maps \eqref{R G1 comps}-\eqref{T G1 comps}, respectively.

The Lax matrices $L_a\,, M_a$ in \eqref{L G1 tensor}, \eqref{M G1 tensor} are obtained by using the tensor product of $\mbox{Mat}_3(\mathbb{F})$ with a representation of $\Gamma(1)$. In particular, first we express the Lax supermatrices $\mathcal{L}_a$ and $\mathcal{M}_a$ in \eqref{Lax} as elements in $\mbox{Mat}_3(\mathbb{F}) \otimes \Gamma(1)$ by expanding their entries in terms of $1$ and $\theta$,  therefore writing them as
\begin{equation}
\mathcal{L}_a({\bm x}, {\bm \chi}) = L_1 \otimes 1+ L_2 \otimes  \theta \,,  \quad  \mathcal{M}_a({\bm x}, {\bm \chi}) = M_1 \otimes 1  + M_2 \otimes \theta\,,
\end{equation}
with 
\begin{equation*}
L_1= 
\begin{pmatrix}
x_1x_2 + a + \lambda &x_1&0\\
x_2&1&0\\
0&0&1
\end{pmatrix}, \; 
M_1=
\begin{pmatrix}
x_2+\lambda&x_1&0\\
\frac{a}{x_1}&0&0\\
0&0&1
\end{pmatrix}, \;
L_2 = M_2=
\begin{pmatrix}
0&0&\chi_1\\
0&0&0\\
\chi_2&0&0
\end{pmatrix}.
\end{equation*}
Then, we represent $1$ and $\theta$ by $2 \times 2$ matrices using the algebra homomorphism $\rho: \Gamma(1) \rightarrow \mbox{Mat}_2(\mathbb{F})$ defined by its action on the basis of the algebra
\begin{equation}
1 \overset{\rho}{\mapsto}
\begin{pmatrix}
1&0\\ 0&1
\end{pmatrix}, \quad 
 \theta \overset{\rho}{\mapsto}
\begin{pmatrix}
0&1\\0&0
\end{pmatrix},
\end{equation}
and thus we obtain $L_a, M_a$ in \eqref{L G1 tensor} and \eqref{M G1 tensor} from
\begin{equation} \label{Lax G1 tensor 1}
L_a = L_1 \otimes \rho(1) + L_2 \otimes \rho(\theta) \, \mbox{ and } \, M_a = M_1 \otimes \rho(1) + M_2 \otimes \rho(\theta)\,.
\end{equation}

\begin{remark}
The matrix refactorisation problems in \eqref{G1 refact} with 
\begin{equation}\label{Lax G1 tensor 2}
\tilde{L}_a = \rho(1) \otimes L_1  + \rho(\theta) \otimes  L_2\,, \quad \tilde{M}_a = \rho(1) \otimes M_1+ \rho(\theta) \otimes M_2 
\end{equation}
are also equivalent to the maps $R_{a,b}, S_{a,b}, T_{a,b}$ in \eqref{R G1 comps}-\eqref{T G1 comps}, since the Lax matrices $L_a, \, M_a$ in \eqref{Lax G1 tensor 1} and $\tilde{L}_a, \, \tilde{M}_a$ in \eqref{Lax G1 tensor 2} are similar under a permutation matrix.
\end{remark}

The expansion in the basis of $\Gamma(1)$ can also be performed for the $\mathcal{I}, \mathcal{J}$ families of invariants in \eqref{R inv}, \eqref{S inv} of maps $\mathcal{R}_{a,b}$ and $\mathcal{S}_{a,b}$ leading to invariants of $R_{a,b}$, $S_{a,b}$. This way, invariants $\mathcal{I}_i$ and $\mathcal{J}_i$, $i=1, \ldots, 5$, are expressed in terms of $1$ and powers of $\theta$. It is interesting to note that, while $\theta^k = 0$ for $k \geq 2$, the coefficients of $\theta^k$ in the expansions can still be invariant quantities.  For example, the invariant $\mathcal{I}_2 = \chi_1\chi_2 + \psi_1\psi_2$ of $\mathcal{R}_{a,b}$ is expressed as $\mathcal{I}_2=I_2\theta^2$ in $\Gamma(1)$. While $\theta^2=0$, it turns out that $I_2$, which involves the commutative coefficients of $\chi_i, \psi_i \in\Gamma(1)$, is an invariant of map $R_{a,b}$ \eqref{R G1 comps}. Following this idea, we obtain the following sets of functionally independent invariants for $R_{a,b}$ and $S_{a,b}$, respectively:
\begin{equation} \label{R inv G1}
\begin{aligned}
I_1 &= x_1x_2 + y_{1}y_{2} \,, \quad I_{2} = \chi_{1} \chi_{2} + \psi_{1} \psi_{2}\,, \quad I_3 = (x_1 \psi_1 - y_1\chi_1)(x_2\psi_2 - y_2 \chi_2) \,,\\
I_{4} &= b x_1x_2+ay_1y_2+x_1x_2y_1y_2+x_1y_2+x_2y_1 \,, 
\end{aligned}
\end{equation}
and
\begin{equation} \label{S inv G1}
\hspace{-0.5cm} 
\begin{aligned}
J_{1} &= y_2+x_1x_2 \,,\quad J_{2} = \chi_{1} \chi_{2} + \psi_{1} \psi_{2} \,, \\
J_{3}&= (a+x_1x_2)\psi_1 \psi_2 + (b+y_2) \chi_1 \chi_2 + (1-x_2y_1)\chi_1 \psi_2 + \left( 1-\frac{b x_1}{y_1} \right)\psi_1  \chi_2\,, \\
J_{4} &= b \frac{x_1}{y_1}+ y_2(a+ x_1x_2) + x_2y_1\,.
\end{aligned}
\end{equation}
We note that the invariants $\mathcal{I}_5$ and $\mathcal{J}_5$ of $\mathcal{R}_{a,b}, \mathcal{S}_{a,b}$ do not produce any invariants for $R_{a,b}, S_{a,b}$.

\begin{remark}
The invariants $I_1, I_4$ of $R_{a,b}$ can also be obtained from the characteristic function of the monodromy matrix $P_{R} = L_b(y_1, y_2, \psi_1, \psi_2)L_a(x_1, x_2, \chi_1, \chi_2)$. Similarly,  $J_1, J_4$ can be obtained from the characteristic function  $P_{S} = M_b(y_1, y_2, \psi_1, \psi_2)L_a(x_1, x_2, \chi_1, \chi_2)$ associated to map $S_{a,b}$.
\end{remark}

Maps $R_{a,b}, S_{a,b}, T_{a,b}$ in \eqref{R G1 comps}-\eqref{T G1 comps} can be written in `triangular form', meaning that each of them can be expressed as the composition of a linear map which acts on $(\chi_1, \chi_2, \psi_1, \psi_2)$ with coefficients being rational functions of $(x_1, x_2, y_1, y_2)$, and a map that acts non-trivially only on the variables $(x_1, x_2, y_1, y_2)$. For example, for $R_{a,b}$ we have the decomposition $R_{a,b} = \bar{R}_{a,b} \circ \hat{R}_{a,b}$, where 
$$
\hat{R}_{a,b}: 
(x_1,x_2,\chi_1,\chi_2,y_1,y_2,\psi_1,\psi_2)\mapsto (x_1,x_2,\xi_1,\xi_2,y_1,y_2,\eta_1,\eta_2)
$$
with 
$$
(\xi_1,\xi_2, \eta_1, \eta_2) = \left( \psi_1+\frac{b-a}{1+x_1y_2}\chi_1,  \chi_1, \chi_2+\frac{a-b}{1+x_1y_2}\psi_2,  \psi_2  \right)\,,
$$
and $\bar{R}_{a,b}$ an extension of the Adler-Yamilov map \cite{adleryamilov1994}
$$
\bar{R}_{a,b}: 
(x_1,x_2,\chi_1,\chi_2,y_1,y_2,\psi_1,\psi_2)\mapsto (u_1,u_2,\chi_1,\chi_2,v_1,v_2,\psi_1,\psi_2)
$$
with
$$
(u_1,u_2,v_1,v_2)=\left( y_1 + \frac{b-a}{1+x_1y_2}x_1, y_2, x_1, x_2+\frac{a-b}{1+x_1y_2}y_2\right ).
$$
In these decompositions of $R_{a,b}, S_{a,b}, T_{a,b}$ the rational maps which act only on the variables $(x_1, x_2, y_1, y_2)$ are extensions of the maps $\mathtt{R}_{a,b}, \mathtt{S}_{a,b}, \mathtt{T}_{a,b}$ in \eqref{YB KRP}, which were shown to be Liouville integrable in \cite{KRP2019}. 

Finally, regarding the dynamical properties of maps $R_{a,b}, S_{a,b}, T_{a,b}$ in \eqref{R G1 comps}-\eqref{T G1 comps},  we show that these maps are measure preserving. This means that for each of them there exists a function $m$ of the dynamical variables such that the Jacobian determinant $J$ of the map can be written as \cite{DiscreteBook}
$$
J = \frac{m(x_1,x_2,\chi_1, \chi_2, y_1,y_2, \psi_1,\psi_2)}{m(u_1, u_2, \xi_1, \xi_2,v_1, v_2, \eta_1, \eta_2)}\,.
$$
It can be verified that maps $R_{a,b}, S_{a,b}, T_{a,b}$ in \eqref{R G1 comps}-\eqref{T G1 comps} are measure preserving with $m(x_1,x_2, \chi_1, \chi_2, y_1, y_2, \psi_1, \psi_2)$ equal to $1$, $\frac{1}{y_1}$ and $\frac{1}{x_1}$, respectively. In particular, the YB map $R_{a,b}$ is volume preserving.

\section{The Grassmann algebra $\Gamma(2)$} \label{sec: n2}

In this section we derive 16-dimensional birational, parametric, entwining YB maps  $\mathsf{R}_{a,b}, \mathsf{S}_{a,b}, \mathsf{T}_{a,b}:\mathbb{F}^{16} \rightarrow \mathbb{F}^{16}$ which act as
\begin{eqnarray}
(x_{11}, x_{12}, x_{21}, x_{22}, \chi_{11}, \chi_{12}, \chi_{21}, \chi_{22}, y_{11}, y_{12}, y_{21}, y_{22}, \psi_{11}, \psi_{12}, \psi_{21}, \psi_{22}) \mapsto \nonumber \\
(u_{11}, u_{12}, u_{21}, u_{22}, \xi_{11}, \xi_{12}, \xi_{21}, \xi_{22}, v_{11}, v_{12}, v_{21}, v_{22}, \eta_{11}, \eta_{12}, \eta_{21}, \eta_{22}) \,. \nonumber
\end{eqnarray}
These maps are obtained from maps $\mathcal{R}_{a,b}, \mathcal{S}_{a,b}, \mathcal{T}_{a,b}$ in \eqref{R comps}-\eqref{T comps} for the case of the Grassmann algebra $\Gamma(2)$, following the ideas presented in Section \ref{sec: n1}. Unlike the case of maps $R_{a,b}, S_{a,b}, T_{a,b}$ obtained in the previous section, the maps presented here are not in `triangular form'. More precisely, although the maps act linearly on $(\chi_{11},\chi_{12}, \chi_{21}, \chi_{22}, \psi_{11}, \psi_{12}, \psi_{21}, \psi_{22})$ with coefficients which are rational functions of only $(x_{11},x_{12}, x_{21}, x_{22}, y_{11},y_{12}, y_{21}, y_{22})$, their action on  $(x_{11},x_{12}, x_{21}, x_{22}, y_{11},y_{12}, y_{21}, y_{22})$ has coefficients which are functions of all the dynamical variables $x_{ij}, \chi_{ij}, y_{ij}, \psi_{ij}$ for $i, j \in \{1,2\}$. Similar to the case of the 8-dimensional birational maps which were derived in Section \ref{sec: n1},  the 16-dimensional maps $\mathsf{R}_{a,b}, \mathsf{S}_{a,b}, \mathsf{T}_{a,b}$ obtained in this section admit a strong Lax triple, are measure preserving, and each of them has a family of invariants.

The elements of the $\Gamma(2)$ Grassmann algebra can be written as linear combinations of  $1, \theta_1, \theta_2, \theta_1\theta_2$ with $\theta_i\theta_j+\theta_j\theta_i = 0$ for $i,j=1,2$. Expressing each of the components of a point $({\bm x}, {\bm \chi})=(x_1, x_2, \chi_1, \chi_2) \in \mathbb{F}_2^{2,2}$  in the basis of $\Gamma(2)$ we have 
\begin{equation}
x_i = x_{i1}  + x_{i2} \theta_1\theta_2 \,, \quad \chi_i = \chi_{i1} \theta_1 + \chi_{i2}\theta_2 \,,
\end{equation}
and for even elements we have
\begin{equation}
x_i^{-1} = \frac{1}{x_{i1}} - \frac{x_{i2}}{x_{i1}^2} \theta_1 \theta_2 \,, \quad \mbox{with} \quad  x_{ij}, \chi_{ij} \in \mathbb{F} \quad \mbox{for} \quad  i, j \in \left\lbrace 1, 2 \right\rbrace. 
\end{equation}
Starting from maps $\mathcal{R}_{a,b}, \mathcal{S}_{a,b}, \mathcal{T}_{a,b}$ that were derived in Section \ref{sec: YB maps}, we express each of their variables as described above.  By comparing coefficients of $1, \theta_1\theta_2$ and $\theta_1, \theta_2$ on both sides of each of the equations \eqref{R comps}-\eqref{T comps}, we obtain sixteen-dimensional maps which we denote by $\mathsf{R}_{a,b}, \mathsf{S}_{a,b}, \mathsf{T}_{a,b}$, respectively. The components of these maps are given by
\begin{equation} \label{RG2}
\mathsf{R}_{a,b}: 
\begin{cases}
u_{11} =y_{11} -  \frac{a-b}{1+x_{11}y_{21}}x_{11} \,, &  v_{11} = x_{11}\,,  \vspace{0.15cm} \\
u_{12} = y_{12} - \frac{(a-b)(x_{12} - x_{11}(x_{11}y_{22}+\chi_{11}\psi_{22} - \chi_{12}\psi_{21}))}{(1+x_{11}y_{21})^2} \,, & v_{12} = x_{12} \,,\\
u_{21} = y_{21} \,, & v_{21} = x_{21}+ \frac{a-b}{1+x_{11}y_{21}}y_{21}  \,,  \vspace{0.15cm}  \\
u_{22} = y_{22} \,, & v_{22} = x_{22} + \frac{(a-b)(y_{22} - y_{21}(x_{12}y_{21}+\chi_{11}\psi_{22}-\chi_{12}\psi_{21}))}{(1+x_{11}y_{21})^2}\,,\\
\xi_{11} = \psi_{11} - \frac{a-b}{1+x_{11}y_{21}}\chi_{11}\,,   & \eta_{11} = \chi_{11} \,,  \\
 \xi_{12} = \psi_{12} - \frac{a-b}{1+x_{11}y_{21}} \chi_{12} \,, & \eta_{12} =\chi_{12} \,,\\
 \xi_{21} =  \psi_{21} \,,  & \eta_{21} = \chi_{21} + \frac{a-b}{1+x_{11} y_{21}} \psi_{21} \,, \\
 \xi_{22} = \psi_{22} \,, & \eta_{22} = \chi_{22} + \frac{a-b}{1+x_{11}y_{21}}\psi_{22} \,,
\end{cases}
\end{equation}
\begin{equation} \label{SG2}
\mathsf{S}_{a,b}: 
\begin{cases}
u_{11} =  \frac{y_{11}^2}{b x_{11}} + \frac{y_{11}(y_{21}-a)}{b} \,,&  v_{11} = x_{11}\,,  \vspace{0.15cm} \\
u_{12} = - \frac{x_{12}y_{11}^2}{b x_{11}^2} + \frac{y_{11} (2 y_{12}+y_{11}(\chi_{12}\psi_{21} + \chi_{11} \psi_{22}))}{bx_{11}} \vspace{0.15cm} \\ \hphantom{123} + \frac{y_{12}(y_{21}-a)}{b} 
+ \frac{y_{11}(y_{22} + \psi_{12}\psi_{21} -\psi_{11}\psi_{22})}{b}\,,  & v_{12} = x_{12} \,,   \\
u_{21} = \frac{b}{y_{11}} \,,& v_{21} = a+ x_{11}x_{21} -\frac{y_{11}}{x_{11}} \,,   \\
u_{22} = - \frac{b y_{12}}{y_{11}^2}\,, & v_{22} = x_{11}x_{22} + x_{12}x_{21} + \frac{x_{12}y_{11}}{x_{11}^2} - \frac{y_{12}}{x_{11}} \\
& \hphantom{123} - \chi_{12}\chi_{21} + \chi_{11}\chi_{22}  \,,\\
\xi_{11} = \psi_{11} - \frac{y_{11}}{x_{11}} \chi_{11} \,,   &  \eta_{11} = \chi_{11} \,,  \\
 \xi_{12} = \psi_{12} - \frac{y_{11}}{x_{11}} \chi_{12} \,, &  \eta_{12} =\chi_{12} \,,\\
 \xi_{21} =  \psi_{21} \,,  &  \eta_{21} = \chi_{21} + \frac{y_{11}}{x_{11}}\psi_{21} \,, \\
 \xi_{22} = \psi_{22} \,, & \eta_{22} = \chi_{22} + \frac{y_{11}}{x_{11}}\psi_{22} \,,
\end{cases}
\end{equation}
and
\begin{equation} \label{TG2}
\mathsf{T}_{a,b}: 
\begin{cases}
u_{11} =  \frac{a}{y_{21}} \,,&  v_{11} = x_{11} \,,  \\
u_{12} = - \frac{a y_{22}}{y_{21}^2} \,,  & v_{12} = x_{12} \,,\\
u_{21} = b - \frac{a}{x_{11}y_{21}} + y_{11}y_{21} \,,& v_{21} = \frac{a+ x_{11}(x_{21}-b)y_{21} }{x_{11}^2 y_{21}}   \,, \vspace{0.15cm}  \\
u_{22} = y_{12}y_{21} + y_{11}y_{22}   \, & v_{22} = - \frac{a (2 x_{12} + x_{11}(\chi_{12}\psi_{21} - \chi_{11}\psi_{22}) )}{x_{11}^3y_{21}}  - \frac{a y_{22}}{x_{11}^2y_{21}^2}\vspace{0.15cm}  \\
\hphantom{123} + \frac{a(x_{12}y_{21}+x_{11}y_{22})}{x_{11}^2 y_{21}^2} + \psi_{11}\psi_{22} - \psi_{12}\psi_{21}\,,   &  \hphantom{123} -\frac{(x_{21}-b)x_{12} }{x_{11}^2}  + \frac{\left(x_{22}+\chi_{12}\chi_{21}-\chi_{11}\chi_{22}   \right)}{x_{11}}  \,, \vspace{0.15cm} \\
\xi_{11} = \psi_{11} + \frac{a \chi_{11}}{x_{11}y_{21}} \,,   &  \eta_{11} = \chi_{11} \,,  \\
 \xi_{12} = \psi_{12} + \frac{a \chi_{12}}{x_{11}y_{21}}  \,, &  \eta_{12} =\chi_{12} \,,\\
 \xi_{21} =  \psi_{21}  \,,  &  \eta_{21} = \chi_{21} -\frac{a \psi_{21}}{x_{11}y_{21}} \,, \\
 \xi_{22} = \psi_{22}  \,, & \eta_{22} = \chi_{22} - \frac{a \psi_{22}}{x_{11} y_{21}}  \,.
\end{cases}
\end{equation}
In what follows we discuss integrability properties of the entwining YB maps given by \eqref{RG2}-\eqref{TG2}, focusing on maps $\mathsf{R}_{a,b}$ and $\mathsf{S}_{a,b}$ due to our previous observation in Remark \ref{ST equiv}.

Invariants of maps $\mathsf{R}_{a,b}$ and $\mathsf{S}_{a,b}$ can be obtained starting from the families of invariants $\mathcal{I}$, $\mathcal{J}$ of $\mathcal{R}_{a,b}$, $\mathcal{S}_{a,b}$, respectively. Following the ideas discussed in the previous section, we first express the variables that appear in  invariants $\mathcal{I}_i$ and $\mathcal{J}_i$ in terms of $1, \theta_1\theta_2$ and $\theta_1, \theta_2$. Then the coefficients of $1$, $\theta_1 \theta_2$ as well as those of powers of $\theta_1$ and $\theta_2$ can lead to invariants for the 16-dimensional maps \eqref{RG2}, \eqref{SG2}. In particular, using only invariants $\mathcal{I}_1, \mathcal{I}_2$ and $\mathcal{I}_4$ from the list \eqref{R inv} we obtain the following functionally independent invariants for map $\mathsf{R}_{a,b}$
\begin{align}\label{R invs G2} 
\mathsf{I}_{1} &= x_{11}x_{21} + y_{11} y_{21} \,, ~ \quad \mathsf{I}_{2} = x_{11}x_{22} + x_{12}x_{21} + y_{11}y_{22} + y_{12}y_{21} \,, \nonumber \\
\mathsf{I}_{3} & = \chi_{11}\chi_{22}  + \psi_{11}\psi_{22}  \,, \quad \mathsf{I}_{4} =  \chi_{12}\chi_{21} +  \psi_{12}\psi_{21}\,, \nonumber \\
 \mathsf{I}_5 &= \chi_{11} \chi_{21} + \psi_{11} \psi_{21} \,, \quad \mathsf{I}_6 = \chi_{12} \chi_{22} + \psi_{12} \psi_{22} \,, \\
 \mathsf{I}_{7} &=  b x_{11}x_{21} + a y_{11}y_{21} +x_{11}y_{21} + x_{21}y_{11} + x_{11}x_{21}y_{11}y_{21}\,, \nonumber \\
\mathsf{I}_{8} &= b(x_{11}x_{22}+x_{12}x_{21} + \chi_{11}\chi_{22} - \chi_{12}\chi_{21}) + a(y_{11}y_{22}+y_{12}y_{21} +\psi_{11}\psi_{22} - \psi_{12}\psi_{21}) \nonumber  \\
& +  x_{11}x_{21} (y_{11}y_{22}+y_{12}y_{21}) + y_{11}y_{21}(x_{11}x_{22} + x_{12}x_{21}) + y_{11}y_{21}(\chi_{11}\chi_{22} -\chi_{12}\chi_{21}) \nonumber \\
& +  x_{11}x_{21}(\psi_{11}\psi_{22} - \psi_{12}\psi_{21}) + x_{11}y_{22} + x_{22}y_{11} + x_{21}y_{12} + x_{12}y_{21} +\chi_{11}\psi_{22}\nonumber  \\ &+ \chi_{22}\psi_{11} -  \chi_{12}\psi_{21} - \chi_{21}\psi_{12} \,. \nonumber
\end{align}
Moreover, expanding the anti-invariants given in Remark \ref{rem: anti-inv} in the basis of $\Gamma(2)$ we obtain the following six anti-invariants of map $\mathsf{R}_{a,b}$
\begin{equation}\label{R anti-invs G2}
\mathsf{A}_{ij} = x_{i 1}\psi_{i j} - y_{i 1} \chi_{i j} ~~ \mbox{and} ~~ \mathsf{B}_{i} = \chi_{i 1} \psi_{i2} - \chi_{i 2} \psi_{i 1} \,, \quad \mbox{for} ~ i, j = 1, 2\,.
\end{equation}
The squares of $\mathsf{A}_{ij}$ and $\mathsf{B}_{i}$, as well as any product of two of them is an invariant of map $\mathsf{R}_{a,b}$. Obviously, not all of invariants \eqref{R invs G2} and those obtained from combinations of the anti-invariants \eqref{R anti-invs G2} form a generating set for the ring of invariants of $\mathsf{R}_{a,b}$, since, for example, the invariants $\mathsf{A}_{ij}^2, \mathsf{B}_k^2$ and $\mathsf{A}_{ij} \mathsf{B}_k$ satisfy the syzygy $(\mathsf{A}_{ij}\mathsf{B}_k)^2=(\mathsf{A}_{ij})^2(\mathsf{B}_k)^2$.
Similarly, we use the invariants $\mathcal{J}_1 - \mathcal{J}_4$ in \eqref{S inv} to find the following functionally independent invariants of map $\mathsf{S}_{a,b}$ 
\begin{align}\label{S G2 inv}
\mathsf{J}_1 &= y_{21} + x_{11}x_{21} \,, \quad \mathsf{J}_2 = y_{22} + x_{11}x_{22} + x_{12}x_{21} + \chi_{11}\chi_{22} - \chi_{12}\chi_{21} \,, \nonumber\\
\mathsf{J}_{3} & = \chi_{11}\chi_{22}  + \psi_{11}\psi_{22}  \,, \quad \mathsf{J}_{4} =  \chi_{12}\chi_{21} +  \psi_{12}\psi_{21}\,, \nonumber\\
 \mathsf{J}_5 &= \chi_{11} \chi_{21} + \psi_{11} \psi_{21} \,, \quad \mathsf{J}_6 = \chi_{12} \chi_{22} + \psi_{12} \psi_{22} \,,\nonumber \\
  \mathsf{J}_7 & = (a + x_{11}x_{21})( \psi_{11}\psi_{22} - \psi_{12}\psi_{21}) + (b + y_{21})(\chi_{11}\chi_{22} - \chi_{12}\chi_{21}) \nonumber \\
 & + (1- x_{21}y_{11})(\chi_{11}\psi_{22} - \chi_{12}\psi_{21}) + \left( 1- \frac{b x_{11}}{y_{11}} \right)( \chi_{22}\psi_{11} - \chi_{21}\psi_{12})\,, \\
 \mathsf{J}_8 &= (a+x_{11} x_{21}) (\psi_{11} \psi_{21} - \psi_{12} \psi_{22}) +(b+y_{21})(\chi_{11}\chi_{21} - \chi_{12}\chi_{22}) \nonumber \\ 
 & +(1-x_{21}y_{11})( \chi_{11}\psi_{21} - \chi_{12}\psi_{22})+ \left(1- \frac{b x_{11}}{y_{11}} \right)( \psi_{11} \chi_{21} -\psi_{12} \chi_{22} ) \,,\nonumber \\
 \mathsf{J}_{9} &= b \frac{x_{11}}{y_{11}} + y_{21}(a+x_{11}x_{21}) +x_{21}y_{11}  \,,\nonumber \\
 \mathsf{J}_{10} &= b \left( \frac{x_{12}}{y_{11}} - \frac{x_{11}y_{12}}{y_{11}^2}\right)  + y_{21}(x_{12}x_{21} + x_{11}x_{22} + \chi_{11}\chi_{22}- \chi_{12}\chi_{21}) + y_{22}(a + x_{11}x_{21}) \nonumber \\
 &  + x_{21}y_{12} + x_{22}y_{11} + \chi_{11}\psi_{22} - \chi_{12}\psi_{21} + \chi_{22}\psi_{11} - \chi_{21}\psi_{12} \,. \nonumber
\end{align}
Additionally, from invariant $\mathcal{J}_5$ we find that $\mathsf{B}_i$, for $i=1,2$, in \eqref{R anti-invs G2} are anti-invariants of map $\mathsf{S}_{a,b}$.

Following the ideas in Section \ref{sec: n1}, we construct a strong Lax triple $(\mathsf{L}_a, \mathsf{M}_a, \mathsf{N}_a)$, with $\mathsf{N}_a \equiv \mathsf{L}_a$, for maps $\mathsf{R}_{a,b}, \mathsf{S}_{a,b}, \mathsf{T}_{a,b}$ \eqref{RG2}-\eqref{TG2}. We start by expressing the Lax matrices with Grassmann variables $\mathcal{L}_a, \mathcal{M}_a$ in \eqref{Lax} in the basis of $\Gamma(2)$ as
\begin{equation} \label{L G2 basis}
\begin{split}
\mathcal{L}_{a}({\bm x}, {\bm \chi})& = \mathsf{L}_1 \otimes 1 + \mathsf{L}_2 \otimes \theta_1\theta_2 + \mathsf{L}_3 \otimes \theta_1 + \mathsf{L}_4 \otimes \theta_2\,, \\
\mathcal{M}_{a}({\bm x}, {\bm \chi}) &= \mathsf{M}_1 \otimes 1 + \mathsf{M}_2 \otimes \theta_1\theta_2 + \mathsf{M}_3 \otimes \theta_1 + \mathsf{M}_4 \otimes \theta_2 \,,
\end{split}
\end{equation}
with the coefficients $\mathsf{L}_i\,, \mathsf{M}_i$ given by 
\begin{equation} \label{L134 G2}
\mathsf{L}_1 = 
\begin{pmatrix}
x_{11} x_{21} + a + \lambda & x_{11} & 0\\
x_{21} & 1 & 0\\
0&0&1
\end{pmatrix}  \!, \;
\mathsf{L}_3 = 
\begin{pmatrix}
0 & 0 & \chi_{11}\\
0 & 0 & 0\\
\chi_{21} & 0 & 0
\end{pmatrix} \!, \;
\mathsf{L}_4= 
\begin{pmatrix}
0 & 0 & \chi_{12}\\
0 & 0 & 0\\
\chi_{22} & 0 & 0
\end{pmatrix} \!,
\end{equation}
\begin{equation}\label{L2 G2}
\mathsf{L}_2=
\begin{pmatrix}
x_{11} x_{22} + x_{12}x_{21} + \chi_{11}\chi_{22} - \chi_{12} \chi_{21} &  x_{12} &  0\\
 x_{22}& 0 &0\\
0 & 0 & 0
\end{pmatrix} \!, 
\end{equation}
and
\begin{equation} \label{M_i G2}
\mathsf{M}_1 = 
\begin{pmatrix}
x_{21} + \lambda & x_{11} & 0 \\
\frac{a}{x_{11}} & 0 & 0 \\
0 & 0 & 1
\end{pmatrix} \!, \quad 
\mathsf{M}_2 = 
\begin{pmatrix}
x_{22} & x_{12} & 0\\
- \frac{a x_{12}}{x_{11}^2} & 0 & 0\\
0 & 0 & 0
\end{pmatrix}\!, \quad \mathsf{M}_3 = \mathsf{L}_3\,, \quad \mathsf{M}_4 = \mathsf{L}_4\,. 
\end{equation}
Then, using the algebra homomorphism $\rho:\Gamma(2) \rightarrow \mbox{Mat}_4(\mathbb{F})$ defined by
\begin{equation}
1 \overset{\rho}{\mapsto} 
\begin{pmatrix}
1&0&0&0\\
0&1&0&0\\
0&0&1&0\\
0&0&0&1
\end{pmatrix}, \;
 \theta_1 \overset{\rho}{\mapsto}
\begin{pmatrix}
0&1&0&0\\
0&0&0&0\\
0&0&0&1\\
0&0&0&0
\end{pmatrix}, \;
 \theta_2 \overset{\rho}{\mapsto} 
\begin{pmatrix}
0&0&1&0\\
0&0&0&-1\\
0&0&0&0\\
0&0&0&0
\end{pmatrix},
\end{equation}
we obtain a strong Lax triple $(\mathsf{L}_a, \mathsf{M}_a, \mathsf{L}_a)$ for maps $\mathsf{R}_{a,b}, \mathsf{S}_{a,b}, \mathsf{T}_{a,b}$ with matrices $\mathsf{L}_a, \mathsf{M}_a$ given by
\begin{equation} \label{Lax G2}
\begin{aligned}
\mathsf{L}_{a}  &= \mathsf{L}_1 \otimes \rho(1) + \mathsf{L}_2 \otimes \rho(\theta_1)\rho(\theta_2) + \mathsf{L}_3 \otimes \rho(\theta_1) + \mathsf{L}_4 \otimes \rho(\theta_2) \\
\mathsf{M}_{a} & = \mathsf{M}_1 \otimes \rho(1) + \mathsf{M}_2 \otimes \rho(\theta_1)\rho(\theta_2) + \mathsf{M}_3 \otimes \rho(\theta_1) + \mathsf{M}_4 \otimes \rho(\theta_2) \,,
\end{aligned}
\end{equation}
and $\mathsf{L}_i, \mathsf{M}_i$, $i=1, \ldots, 4$ given in \eqref{L134 G2}, \eqref{L2 G2} and \eqref{M_i G2}. The 16-dimensional maps $\mathsf{R}_{a,b}$ and $\mathsf{S}_{a,b}$ in \eqref{RG2}, \eqref{SG2} are equivalent to the matrix refactorisation problems of the $12 \times 12$ Lax matrices \eqref{Lax G2}, similar to the refactorisation problems \eqref{G1 refact} in Section \ref{sec: n1}.

Finally, similar to the $\Gamma(1)$ case, each of the maps $\mathsf{R}_{a,b}, \mathsf{S}_{a,b}, \mathsf{T}_{a,b}$ in \eqref{RG2}-\eqref{TG2} admits an invariant measure $m$. These measures are $m=1$ for $\mathsf{R}_{a,b}$, $m=\frac{1}{y_{11}^2}$ for $\mathsf{S}_{a,b}$, and $m=\frac{1}{x_{11}^2}$ for $\mathsf{T}_{a,b}$. In particular, we observe that the commutative consequences of map $\mathcal{R}_{a,b}$ for $n=1$ and $n=2$, that is maps $R_{a,b}$ and $\mathsf{R}_{a,b}$, are volume preserving maps. We conjecture that in $\Gamma(n)$ the map $\mathcal{R}_{a,b}$ is volume preserving for every $n$, while the maps $S_{a,b}$ and $T_{a,b}$ preserve measures of the form $y_{11}^{-n}$ and $x_{11}^{-n}$, respectively, with $x_{11}$ and $y_{11}$ defined similar to the cases $n=1$ and $n=2$.

\section{Conclusions}\label{sec: concl}

We have constructed birational maps $\mathcal{R}_{a,b},~ \mathcal{S}_{a,b},~ \mathcal{T}_{a,b}$ with Grassmann variables given in \eqref{R comps}-\eqref{T comps}, which satisfy the set-theoretical entwining YB equation \eqref{eYB_eq1}. These maps admit a strong Lax triple, which we used to derive invariants for the maps. The invariants that we find for map $\mathcal{R}_{a,b}$ are all polynomial, while those of maps $\mathcal{S}_{a,b}$ and $\mathcal{T}_{a,b}$ are Laurent polynomials with negative powers appearing only on even variables of the Grassmann algebra. Reversing this point of view, one could make connections with non-commutative algebraic geometry by viewing the maps as birational automorphisms of non-commutative algebraic varieties.

In Sections \ref{sec: n1} and \ref{sec: n2} we have shown how a hierarchy of birational entwining YB maps in dimensions $2^{n+2}$, where $n$ is the order of the Grassmann algebra, can be obtained. The case $n=0$, i.e. when there are no fermionic variables, was studied in \cite{KRP2019}. Here, we considered in detail the cases where $n=1, 2$, thus obtaining birational maps over $\mathbb{F}^8$ and $\mathbb{F}^{16}$ that satisfy the entwining Yang-Baxter equation. We have derived sufficient number of independent invariants of these maps to claim their Liouville integrability, however, we have not yet been able to find a Poisson structure. Nevertheless, there are indications which point towards the integrability of the maps with commutative variables. We have found that these maps are measure preserving, and some preliminary numerical experiments show no existence of chaos. Moreover, we have written each of the $8$-dimensional maps of Section \ref{sec: n1} as a composition of a Liouville integrable map with a linear map. More insight regarding the integrability of the maps could be gained using other methods, such as singularity confinement or algebraic entropy. All $2^{n+2}-$dimensional maps arise from refactorization problems of Lax matrices, which we present for $n=1$ and $n=2$. It would be interesting to study the associated transfer maps \'a la Veselov \cite{Veselov2003} for each $n$. Finally, defining appropriately the concept of Liouville integrability in the setting of Grassmann-extended entwining YB maps is an interesting open problem.

\section*{Acknowledgements}
We would like to thank Dr S. Konstantinou-Rizos for very fruitful discussions and suggestions in the initial stages of this work.  We would also like to thank the organisers of the ISLAND VI:  Dualities and Symmetries in Integrable Systems for the inspiring atmosphere of the conference, during which part of this work was completed. 


\appendix

\section{Proof of Theorem \ref{RST maps}} \label{proof EYB}
The proof of Theorem \ref{RST maps} is based on Proposition 3.1 in \cite{KP2011}. To prove that the maps \eqref{R comps}-\eqref{T comps} satisfy the entwining YB equation \eqref{eYB_eq1} we have to show that the equation 
\begin{equation} \label{trifact Append}
\mathcal{L}_a(\bm x,\bm \chi) \mathcal{M}_b  (\bm y,\bm \psi) \mathcal{L}_c (\bm z,\bm \zeta) = \mathcal{L}_a(\bm x',\bm \chi') \mathcal{M}_b  (\bm y',\bm \psi') \mathcal{L}_c (\bm z',\bm \zeta')\,, 
\end{equation}
with $\mathcal{L}_a(\bm x,\bm \chi)$ and $\mathcal{M}_a(\bm x,\bm \chi) $ given in \eqref{Lax}, implies $(\bm x,\bm \chi)=(\bm x',\bm \chi'), \, (\bm y,\bm \psi)= (\bm y',\bm \psi')$ and $(\bm z,\bm \zeta)=(\bm z',\bm \zeta')$. Here all the ordered pairs, e.g. $(\bm{ x, \chi})= (x_1,x_2,\chi_1,\chi_2)$,  are in $\mathbb{F}_n^{2,2}$. 

\begin{proof}
We use the standard notation of $e_{ij}$ denoting the matrix with $1$ in the $(i,j)$ entry and $0$ elsewhere. Then the Lax matrices $\mathcal{L}_a(\bm x, \bm \chi)$ and $\mathcal{M}_b(\bm y, \bm \psi)$ are of the form
$$
\mathcal{L}_a(\bm x, \bm \chi) = \lambda e_{11}+A_a(\bm x, \bm \chi), \quad \mathcal{M}_b(\bm y, \bm \psi)=\lambda e_{11}+B_b(\bm y, \bm \psi)
$$
where 
$$
A_a(\bm x, \bm \chi)=
\begin{pmatrix}
    x_1x_2+\chi_1\chi_2+a & x_1 & \chi_1 \\
    x_2 & 1 & 0 \\
    \chi_2 & 0 & 1
\end{pmatrix} \quad \mbox{and} \quad
B_b(\bm y, \bm \psi)= 
\begin{pmatrix}
    y_2 & y_1 & \psi_1 \\
    \frac{b}{y_1} & 0 & 0 \\
    \psi_2 & 0 & 1
\end{pmatrix}.
$$
For simplicity we also introduce the notation
$$
X_a:= x_1x_2+\chi_1\chi_2+a, \quad Z_c:=z_1z_2+\zeta_1\zeta_2+c.
$$
Moreover, we introduce the operators $L_{e_{11}}$ and $R_{e_{11}}$ acting on $\mathrm{M}_{2,1}$ by left and right multiplication by $e_{11}$, respectively.  Since, $e_{11}^2=e_{11}$, the operators $L_{e_{11}}$ and $R_{e_{11}}$ are projections. More precisely, we have
that 
$$
L_{e_{11}}(P)=e_{11}P=\sum_{j=1}^3p_{1j}e_{1j}, \quad R_{e_{11}}(P)=Pe_{11}=\sum_{i=1}^3p_{i1}e_{i1},
$$
for any matrix $P=(p_{ij})\in \mathrm{M}_{2,1}$. It also follows that $L_{e_{11}}\circ R_{e_{11}}(P)=R_{e_{11}}\circ L_{e_{11}}(P)=p_{11}e_{11}$.

We denote the left hand side of $\eqref{trifact Append}$ by $\mathcal{Q}(\lambda)$ and expand it in powers of $\lambda$. We obtain that
$$
\mathcal{Q}(\lambda)=\lambda^3e_{11}+\lambda^2\mathcal{Q}_2+\lambda\mathcal{Q}_1+\mathcal{{Q}}_0,
$$
with $\lbrace\mathcal{Q}_i\rbrace_{i=0}^2$ given by the following expressions 
$$
\begin{array}{rcl}
   \mathcal{Q}_2  & =&  L_{e_{11}}\circ R_{e_{11}}\bigl(B_b(\bm y, \bm \psi)\bigr)+R_{e_{11}}\bigl(A_a(\bm x, \bm \chi)\bigr)+L_{e_{11}}\bigl(A_c(\bm z, \bm \zeta)\bigr) , \\
   \mathcal{Q}_1  & =&  A_a(\bm x,\bm \chi)R_{e_{11}}\bigl(B_b(\bm y,\bm \psi)\bigr) + A_a(\bm x, \bm \chi)L_{e_{11}}\bigl(A_c(\bm z,\bm \zeta)\bigr)+L_{e_{11}}\bigl(B_b(\bm y, \bm \psi)\bigr)A_c(\bm z, \bm \zeta) \\
   \mathcal{Q}_0  & =&  A_a(\bm x, \bm \chi)B_b(\bm y, \bm \psi)A_c(\bm z, \bm \zeta).
\end{array}
$$
Expanding the right hand side of equation \eqref{trifact Append} in $\lambda$, we obtain similar expressions which we denote by $\lbrace \mathcal{Q}_i'\rbrace_{i=0}^2$. It follows that \eqref{trifact Append} implies the matrix equations $\mathcal{Q}_i=\mathcal{Q}_i'$, for $i=0,1,2$.

Matrix equation $\mathcal{Q}_2=\mathcal{Q}_2'$ results in nontrivial equations only for the entries in the first column and the first row. Comparing the coefficients of the matrices $e_{21}$ and $e_{31}$ in $\mathcal{Q}_2$ and $\mathcal{Q}_2'$ gives
$$
x_2=x_2', \quad \chi_2=\chi_2'.
$$
Similarly, from the coefficients of $e_{12}$ and $e_{13}$ we obtain
$$
z_1=z_1', \quad \zeta_1=\zeta_1'.
$$

In the matrix equation $\mathcal{Q}_0=\mathcal{Q}_0'$ we focus on the equations that we obtain from the coefficients of $e_{22}$, $e_{23}$ and $e_{32}$. The equation that corresponds to $e_{22}$ reads
$$
\left(x_2y_2+\frac{b}{y_1}\right)z_1+x_2y_1=\left(x_2y'_2+\frac{b}{y'_1}\right)z_1+x_2y'_1,
$$
where we have used the fact that $x_2=x_2'$ and $z_1=z_1'$. The above equation is polynomial in $z_1$ and therefore it implies that
$$
y_1=y_1', \quad y_2=y_2'.
$$
Similarly, using the equations obtained from the coefficients of $e_{23}$ and $e_{32}$ we have that
$$
\psi_1=\psi_1', \quad \psi_2=\psi_2'.
$$

From the coefficient of $e_{12}$ in $\mathcal{Q}_1=\mathcal{Q}_1'$ we obtain the equation
$$
X_az_1+y_2z_1+y_1=X'_az_1+y_2z_1+y_1 \,,
$$
where we have used the previously obtained equalities between primed and non-primed variables. The latter equation implies that $X_a=X_a'$. Similarly, from the coefficients of $e_{21}$ of the same matrix equation we obtain $Z_c=Z_c'$.

The coefficients of $e_{21}$ and $e_{31}$ in matrix equation $\mathcal{Q}_0=\mathcal{Q}_0'$ give two equations involving $z_2,z'_2$ and $\zeta_2, \zeta'_2$. Using the fact that $Z_c=Z_c'$ these two equations can be written as the following homogeneous system:
$$
\begin{pmatrix}
    y_1 & \psi_1 \\
    \chi_2y_1 & 1+\chi_2\psi_1
\end{pmatrix}
\begin{pmatrix}
    z_2'-z_2 \\
    \zeta_2'-\zeta_2
\end{pmatrix}=
\begin{pmatrix}
    0\\
    0
\end{pmatrix}.
$$
Since the supermatrix of coefficients of the above system is invertible, it follows that
$$
z_2=z_2', \quad \zeta_2=\zeta_2'. 
$$
Finally, from the equations that correspond to the elements $e_{12}$ and $e_{13}$, and using the fact that $X_a=X_a'$, we obtain a similar linear system that results to the remaining equalities
$$
x_1=x_1', \quad \chi_1=\chi_1'.
$$
\end{proof}

\newpage


\end{document}